\documentclass[12pt,a4]{article}

\usepackage{natbib}
\usepackage{amsthm,amsfonts,epsfig,color}
\usepackage{mathptmx}
\usepackage{bm}
\usepackage{amsmath}
\usepackage[T1]{fontenc}
\usepackage{graphicx}
\usepackage{graphics}
\usepackage{lscape}
\usepackage{pdflscape}
\usepackage{amssymb,srcltx}
\usepackage{booktabs, subfigure}
\usepackage{setspace}
\usepackage{hyperref}
\usepackage{accents}
\usepackage{multirow}
\usepackage{bbm}

\newcommand{\Ex}{\textrm{E}}

\newcommand{\iid}{\stackrel {{\rm iid}}{\sim}}
\newcommand{\ind}{\stackrel {{\rm ind}}{\sim}}

\newcommand{\sumas}{\sum^N_{i=1}}
\newcommand{\ii}{i=1,\ldots,N}

\newcommand{\x}{\mathbf{x}}
\newcommand{\yp}{\mathbf{y}}
\newcommand{\Y}{\mathbf{Y}}
\newcommand{\D}{\mathbf{D}}

\newcommand{\z}{\mathbf{z}}
\newcommand{\Z}{\mathbf{Z}}

\newcommand{\A}{\textrm{A}}
\newcommand{\Ae}{\mathbf{A}}

\newcommand{\X}{\mathbf{X}}

\newcommand{\br}{\mathbf{r}}

\newcommand{\s}{\mathbf{s}}

\newcommand{\y}{\mathbf{y}}

\newcommand{\M}{\mathbf{M}}

\newcommand{\bq}{\mathbf{q}}

\newcommand{\bV}{\mathbf{V}}

\newcommand{\W}{\mathbf{W}}

\newcommand{\T}{\mathbf{T}}

\newcommand{\bS}{\mathbf{S}}
\newcommand{\bs}{\mathbf{s}}

\newcommand{\balpha}{\mbox{${\bm \alpha}$}}
\newcommand{\bmu}{\mbox{${\bm \mu}$}}
\newcommand{\bphi}{\mbox{${\bm \phi}$}}

\newcommand{\bSigma}{\mbox{${\bm \Sigma}$}}

\newcommand{\bepsilon}{\mbox{${\bm \epsilon}$}}
\newcommand{\bLambda}{\mbox{${\bm \Lambda}$}}
\newcommand{\bbeta}{\mbox{${\bm \beta}$}}
\newcommand{\btheta}{\mbox{${\bm \theta}$}}

\newcommand{\bDelta}{\mbox{${\bm \Delta}$}}
\newcommand{\bdelta}{\mbox{${\bm \delta}$}}
\newcommand{\bPsi}{\mbox{${\bm \Psi}$}}

\newcommand{\be}{\mathbf{b}}

\newcommand{\bOmega}{\mbox{${\bm \Omega}$}}

\newcommand{\qhatk}{\widehat{Q}^{(k)}}
\newcommand{\ui}{\widehat{u}_i}

\newcommand{\uki}{\widehat{u}_i^{(k)}}

\newcommand{\umr}{\mathbf{1}_r}

\newtheorem{lema}{Lemma}
\newtheorem{proposition}{Proposition}

\newtheorem{definition}{Definition}

\begin{document}
\title{Canonical fundamental skew-t linear mixed models}
\author{
\small Fernanda Lang Schumacher\textsuperscript{a}
\and \small  Larissa Avila Matos\textsuperscript{a}\thanks{Corresponding author.  Address for correspondence: Departamento de Estat\'{\i}stica,  Rua Sérgio Buarque de Holanda, 651. CEP 13083-859. Campinas, SP, Brazil.  
e-mail adresses: \texttt{fernandalschumacher@gmail.com} (F. L. Schumacher),  
\texttt{larissam@ime.unicamp.br} (L. A. Matos),
\texttt{celsoromulo@ufam.edu.br} (C. R. B. Cabral)} 
\and \small  Celso R\^omulo Barbosa Cabral\textsuperscript{b}   \\
{ \em \small \textsuperscript{a} Departamento de Estat\'{\i}stica, Universidade Estadual de Campinas, Brazil}\\ 
{ \em \small\textsuperscript{b}Departamento de Estat\'{\i}stica, Universidade Federal do Amazonas, Brazil}  \vspace*{0.1cm}
}
\date{}
\maketitle
\begin{abstract}
\noindent In clinical trials, studies often present longitudinal data or clustered data. 
These studies are commonly analyzed using linear mixed models (LMMs), usually considering Gaussian assumptions for random effect and error terms. Recently, several proposals extended the restrictive assumptions from traditional LMM by more flexible ones that can accommodate skewness and heavy-tails and consequently are more robust to outliers. This work proposes a canonical fundamental skew-t linear mixed model (ST-LMM), that allows for asymmetric and heavy-tailed random effects and errors and includes several important cases as special cases, which are presented and considered for model selection. For this robust and flexible model, we present an efficient EM-type algorithm for parameter estimation via maximum likelihood, implemented in a closed form by exploring the hierarchical representation of the ST-LMM. In addition, the estimation of standard errors and random effects is discussed. The methodology is illustrated through an application to schizophrenia data and some simulation studies.
\vspace*{0.5cm}\\
\noindent {\bf Keywords:} Heavy-tails; Longitudinal data; Robust models; Skewed distributions
\end{abstract}

\section{Introduction}

Linear mixed models (LMM) are commonly used to model data that present a natural hierarchical structure because they flexibly model the within-subject correlation \citep{PinheiroBates2000}. This kind of structure appears when a variable of interest is repeatedly measured for several subjects (or cluster units, in general), which is frequently the case in clinical trials.

For mathematical convenience, it is usually assumed that both random effect and error follow Gaussian distributions. Nevertheless, these restrictive assumptions may result in a lack of robustness against departures from the normal distribution and invalid statistical inferences, especially when the data show heavy tails and skewness \citep[see, e.g.,][]{drikvandi2017diagnosing,drikvandi2019nonlinear}.
Several approaches have been considered in the literature to replace the normal assumptions of LMM with more flexible distributions. For example, \cite{Pinheiro01} proposed a multivariate t linear mixed model (T-LMM), and \cite{rosa2003robust} considered the thick-tailed class of normal/independent (NI) distributions from a Bayesian framework.

Accounting for skewness, \cite{ArellanoLachos2005} proposed a skew-normal linear mixed model (SN-LMM) based on the \emph{classic skew-normal} (SN) distribution introduced by \cite{azzalini96}, and  \cite{ho2010robust} proposed a skew-t linear mixed model (ST-LMM) based on the \emph{classic skew-t} (ST) distribution introduced by \cite{azzalini2003distributions}.
More generally, \cite{Lachos_Ghosh_Arellano_2010} proposed robust parametric modeling of LMM based on skew-normal/independent (SNI) distributions, where random effects follow an SNI distribution and within-subject errors follow a normal/independent (NI) distribution, and consequently observed responses follow an SNI distribution, and they define what they call the skew-normal/independent linear mixed model, and \cite{schumacher2021scale} extended the SNI-LMM by considering within-subject serial dependence and developing additional tools for model evaluation.

In general, SN and ST distributions extend the normal and Student-t distributions by introducing additional parameters regulating skewness. Some extensions and unifications of these distributions are carefully surveyed in works such as \cite{azzalini2005} and \cite{Arellano_Azzalini_2006}.  For further information, we refer to the book edited by \cite{Genton.Beyond} and the more recent one by \cite{Azzalini.Capitanio.SNBook}.  
The classic formulation of SN and ST distributions were used successfully in many other works, such as \cite{Pyne.etal.2009},  \cite{Lachos_Ghosh_Arellano_2010}, \cite{Lachos_Vilca_Bolfarine_Ghosh}, \cite{Cabral_Lachos_Madruga2012}, \cite{Cabral2012}, and \cite{Cabral.Lachos.Zeller.2014}. Furthermore, another popular version of these two skewed distributions was defined by \cite{Sahu_Dey_Marcia}, which will be called the SDB-SN and SDB-ST in this work. This proposal was applied in works such as \cite{lin2009maximum}, \cite{Lin2010}, and \cite{jara2008linear}.

Recently, \cite{Lee.McLachlan.2016} proposed finite mixtures of a generalization of the classic and SDB-ST distributions, called canonical fundamental skew-t distributions (CFUST), which are special cases of the fundamental skew distributions defined by \cite{Arellano_Genton_2005}.  Using simulation studies, \cite{Lee.McLachlan.2016} showed that the CFUST distribution outperforms the classic and SDB-ST distributions in mixture models. 
In this regard, this work aims to extend the classic ST-LMM by considering the CFUST distribution used in \cite{Lee.McLachlan.2016}, by developing ML estimation based on an EM-type algorithm. This formulation, which will be hereafter simply called ST distribution, enables a more flexible skewness structure at the cost of a higher number of parameters to be estimated.

The rest of this manuscript is organized as follows. Section \ref{sec:skewed} introduces the formulation of the skewed distributions that are considered in this work. Section \ref{sec:model} defines the ST-LMM and discusses its ML estimation via an EM-type algorithm. Section \ref{sec:sim} presents some simulation studies conducted to evaluate the empirical performance of the ST-LMM and the effect of initial values, as well as an illustrative study to exemplify the flexibility of the model. In Section \ref{sec:applic}, the methodology is applied to a schizophrenia data set. Finally, Section \ref{sec:final} discusses some final remarks.

\section{Skewed distributions}\label{sec:skewed}

Let $\X \sim \mathrm{N}_p(\bmu,\bSigma)$ denote a $p$-dimensional random vector following a normal distribution with mean vector $\bmu$ and covariance matrix $\bSigma$, and let $\phi_p(\cdot \mid \bmu,\bSigma)$ and $\Phi_p(\cdot \mid \bmu,\bSigma)$ be the respective density and cumulative distribution function. When $\bmu=\mathbf{0}_p$ and $\bSigma=\mathbf{I}_p$ (the null $p$-dimensional vector and  the $p \times p$ identity matrix, respectively), we simplify the notation to $\phi_p(\cdot)$ and $\Phi_p(\cdot)$, and when $p=1$, we use the notation $\phi(\cdot \mid \mu,\sigma^2)$ and $\Phi(\cdot \mid \mu,\sigma^2)$.

 Suppose $\X \sim \mathrm{N}_p(\bmu,\bSigma)$, then for a given Borel set $A$, we say that the distribution of $\Y=(\X \mid \X \in A) $ is  a \emph{truncated normal distribution on $A$}, denoted by $\Y \sim \mathrm{TN}_p(\bmu,\bSigma,A)$, whose density is given by 
$$
f_{\Y}(\y) = \frac{1}{P(\X \in A)} \phi_{p}(\y \mid \bmu,\bSigma) \mathbbm{1}_{A}(\y),
$$
where $\mathbbm{1}_{A}(\cdot)$ is the indicator function of $A$. As a particular case, consider $\bSigma = \mathbf{I}_p $ and $A = \{\x; \x > \bmu\}$, then all elements of the vector $\X$ are independent random variables and $P(\X>\bmu)=2^p$. 
Now, considering $\X \sim \mathrm{N}_p(\mathbf{0},\bSigma) $, we say that the distribution of $\Y=|\X|$ is a $p$-dimensional half-normal with scale matrix $\bSigma$, where $|\X|=(|X_1|,\ldots,|X_p|)^\top$.

  First, we will define the version of the skew-normal distribution that will be used in this work. It is a special case of the  fundamental skew-normal distribution defined by \citet{Arellano_Genton_2005}. The presentation below is based on this work, and all the proofs can be found there. Some of them are reproduced below, and some others are skipped.  
\begin{definition} \label{def CFUSN}
	Let   $\X_0 \sim \textrm{N}_q(\mathbf{0},\mathbf{I}_q)$ and 
	$\X_1 \sim \textrm{N}_p(\bmu,\bOmega)$ be independent, where  $\bOmega$ is positive definite, and let $\bDelta$ be a $p \times q$ matrix.  We say that the distribution of 
$$
\Y =  \bDelta  |\X_0| +   \X_1
$$	
is skew-normal with location vector $\bmu$, shape matrix $\bDelta$, and  scale matrix $\bOmega$.   We use the notation $\Y \sim \textrm{SN}_{p,q}(\bmu,\bOmega,\bDelta)$. 
\end{definition}

Observe that when $\bDelta = \mathbf{0}$, then $\Y \sim \mathrm{N}_p(\bmu,\bOmega)$.  A trivial but relevant consequence  of this definition is that affine transformations of SN distributions are still SN, as stated in Proposition \ref{prop affine transf SMSN}.
\begin{proposition} \label{prop affine transf SMSN}
	Let  $\Ae$ be an $m \times p$ matrix with rank $m$,  $\be$ be a vector of length $m$, and $\Y \sim  SN_{p,q}(\bmu,\bOmega,\bDelta)$. Then,  $\Z= \Ae\Y+\be \sim  SN_{m,q}(\Ae\bmu  +\be,\Ae \bOmega \Ae',\Ae \bDelta)$. 
\end{proposition}
In particular, marginal distributions are also SN. Thus, if $\Y \sim \textrm{SN}_{p,q}(\bmu,\bOmega,\bDelta)$ and considering the partition 	
\begin{equation}\label{eqn partition SN}
\Y=(\Y_1^\top,\Y_2^\top)^\top,\,\,\,  \mbox{where} \,\,\,  \Y_1: p_1 \times 1 \,\,\,  \mbox{and} \,\,\,   \Y_2: p_2 \times 1, \,\,\, \mbox{with} \,\,\,  p_1+p_2=p,
\end{equation}
then, for $\Ae = (\mathbf{I}_{p_1}\,\,\, \mathbf{0}_{p_1 \times p_2})$, we have $\Y_1 = \Ae \Y$. Matrix $\Ae$ induces similar partitions on $\bmu$, $\bOmega$ and $\bDelta$, given by $\bmu=(\bmu_1^\top,\bmu_2^\top)^\top$, $\bDelta=(\bDelta_1^\top,\bDelta_2^\top)^\top$,  $\bOmega=(\bOmega_{ij})$, $i,j=1,2$, where  $\bmu_1:p_1 \times 1$, $\bOmega_{11}:p_1 \times p_1$ and $\bDelta_1:p_1 \times q$. By Proposition \ref{prop affine transf SMSN}, we have $\Y_1 \sim \textrm{SN}_{p_1,q}(\bmu_1,\bOmega_{11},\bDelta_{1})$. An analogous result is true for $\Y_2$.

\begin{proposition}
	Let $\Y \sim \textrm{SN}_{p,q}(\bmu,\bOmega,\bDelta)$. Then, the density of \ $\Y$ is given by
	\begin{equation} \label{eqn pdf CFUSN}
	f_\Y(\y)= 2^q \phi_p(\y\mid \bmu,\bSigma) \Phi_q(\bDelta^\top\bSigma^{-1}(\y-\bmu) \mid \mathbf{0},\bLambda),
	\end{equation}
	where 
	\begin{equation} \label{eqn siglambda}
	\bSigma = \bOmega+ \bDelta \bDelta^\top \, \, \,   \mbox{and}  \,\, \, \bLambda = \mathbf{I}_q - \bDelta^\top \bSigma^{-1} \bDelta.
	\end{equation} 
\end{proposition}
\begin{proof}
	Define the random vector
	 \begin{equation} \label{eqn def T}
	\T=\bDelta\X_0+\X_1, 
	\end{equation} 
	where $\X_0$ and $\X_1$ are given in Definition \ref{def CFUSN}. Then,  $(\T|\X_0>\mathbf{0})$ has the same distribution of $\Y= \X_1 + \bDelta |\X_0|$.  Thus, it is enough to find the distribution of $(\T|\X_0>\mathbf{0})$, which is $f_\Y(\y)=P(\X_0>\mathbf{0})^{-1}P(\X_0>\mathbf{0}|\T=\y)f_{\T}(\y)$. 
	It can be shown that 
	\begin{equation} \label{eqn stoch rep stand SN}
	\begin{pmatrix}
	\X_0\\
	\T
	\end{pmatrix} \sim \mathrm{N}_{q+p} 
	\left[ \begin{pmatrix}\mathbf{0} \\ \bmu \end{pmatrix},
	\begin{pmatrix}
	\mathbf{I} & \bDelta^\top\\
	\bDelta & \bOmega+ \bDelta \bDelta^\top
	\end{pmatrix}
	\right],
	\end{equation}
	which implies $(\X_0|\T=\y) \sim \mathrm{N}_q(\bDelta^\top\bSigma^{-1}(\y-\bmu),\mathbf{I}-\bDelta^\top\bSigma^{-1}\bDelta)$.  The result follows immediately.  
\end{proof}

\noindent {\bf Remarks}
 \begin{enumerate}
 \item Observe that the skew-normal distribution of \cite{Sahu_Dey_Marcia} is a particular case of the definition above when $p=q$ and $\bDelta$ is a diagonal matrix. Moreover, when $\X_0$ is univariate, we have the classic skew-normal distribution used, for example, in \cite{Lachos_Ghosh_Arellano_2010} and \cite{Cabral2012}. 
\item Observe that the relation $(\bOmega,\bDelta) \to (\bOmega+\bDelta \bDelta^\top,\bDelta)$ induces a 1-1 parameterization, that is, we could parameterize the distribution in terms of $\bSigma$ and use the notation $\Y \sim \textrm{SN}_{p,q}(\bmu,\bSigma,\bDelta)$ with density given in \eqref{eqn pdf CFUSN}. Besides this, if we define the parameterization $\bDelta^*=\bSigma^{-1/2}\bDelta$,  where $\bSigma^{-1/2}$ is a square root of $\bSigma^{-1}$,   the pdf of  $\Y \sim \textrm{SN}_{p,q}(\bmu,\bSigma,\bDelta^*)$ is
\begin{equation*} 
f_{\Y}
(\y)= 2^q \phi_p(\y\mid \bmu,\bSigma) \Phi_q(\bDelta^{*\top}\bSigma^{-1/2}(\y-\bmu) \mid \mathbf{0},\bLambda),
\end{equation*}
where $ \bLambda = \mathbf{I}_q - \bDelta^\top \bSigma^{-1} \bDelta = \mathbf{I}_q-\bDelta^{*\top} \bDelta^{*}$. This parameterization is used in \citet{Arellano_Genton_2005} -- see equation (2.11) and page 109 for the pdf and moment generating function. Unless stated explicitly, we use the parameterization $(\bmu,\bOmega,\bDelta)$.  
\end{enumerate}

The  mean vector and covariance matrix of a random vector with SN distribution are given in Proposition \ref{prop Exp CFUSN}. The proof is a direct consequence of Definition \ref{def CFUSN} and the fact that if a random variable has a univariate half-normal distribution with scale parameter 1, then its mean value is $\sqrt{2/\pi}$.  We define the $q$-dimensional vector of ones by $\mathbf{1}_q=(1,\ldots,1)^\top$.
\begin{proposition} \label{prop Exp CFUSN}
Let $\Y \sim \textrm{SN}_{p,q}(\bmu,\bOmega,\bDelta)$. Then,
$$
\mathrm{E}(\Y) =  \bmu +  \sqrt{\frac{2}{\pi}} \bDelta \mathbf{1}_q \,\,\, \mbox{and} \,\,\, \mathrm{Var}(\Y) = \bSigma - {\frac{2}{\pi}} \bDelta \bDelta^\top.  
$$
\end{proposition}

Before defining the skew-t distribution, we enunciate a  result  regarding marginal and conditional distributions of a random vector with Student-t distribution that will be helpful to obtain results for the skew-t distribution similar to the ones presented to the SN distribution.  The proof can be found in  \citet{Arellano.Bolfarine.1995} (see also \citet[Theorem ~3.8]{Fang.Kotz.Ng.1990}). 
\begin{lema}\label{eqn t closed}
	Let  $\Y \sim \textrm{\em t}_p(\bmu,\bSigma,\nu)$ and consider the partition given in (\ref{eqn partition SN}) and the induced  partitions of   $\bmu$ and $\bSigma$. It can be shown that:
	\begin{itemize}
		\item[(i)] $\Y_1 \sim \textrm{\em
			t}_{p_1}(\bmu_1,\bSigma_{11},\nu)$,
		\item[(ii)] $\Y_2|\Y_1=\y_1 \sim \textrm{\em t}_{p_2}(\bmu_{2.1},\widetilde{\bSigma}_{22.1},\nu+p_1)$,
	\end{itemize}
	where \,
	  $\bmu_{2.1}=\bmu_2+\bSigma_{21}\bSigma_{11}^{-1}(\y_1-\bmu_1)$ and \, $\widetilde{\bSigma}_{22.1}=\dfrac{\nu+ d_1(\y_1)}{\nu+p_1}\ \bSigma_{22.1}$,  with\, $ \bSigma_{22.1}=\bSigma_{22}-\bSigma_{21} \bSigma_{11}^{-1} \bSigma_{12}$ and\, $d_1(\y_1)=(\y_1-\bmu_1)^\top \bSigma_{11}^{-1} (\y_1-\bmu_1)$.
\end{lema}

\begin{definition} \label{def CFSUT}
Let $\X \sim \textrm{SN}_{p,q}(\mathbf{0},\bOmega,\bDelta)$ and $U \sim \textrm{Gamma}(\nu/2,\nu/2)$ be independent, where $\textrm{Gamma}(a,b)$ denotes a Gamma distribution with mean $a/b$ and variance  $a/b^2$, with $b>0$. Let $\bmu$ be a vector of constants of length $p$. We say that the distribution of $\Y = \bmu + U^{-1/2} \X$ is  skew-t (ST) with location vector $\bmu$, scale matrix $\bOmega$, shape matrix $\bDelta$, and $\nu$ degrees of freedom. We use the notation $\Y \sim ST_{p,q}(\bmu,\bOmega,\bDelta,\nu)$. 	
\end{definition}

In what follows, $\mathrm{t}_{p}(\cdot \mid \bmu,\bSigma,\nu)$ and ${T}_{p}(\cdot \mid \bmu,\bSigma,\nu)$ respectively denote the density and the cumulative distribution of the $p$-variate Student-t distribution with location vector $\bmu$, scale matrix $\bSigma$ and $\nu$ degrees of freedom.
\begin{proposition} \label{prop CFUST density}
If  $\Y \sim ST_{p,q}(\bmu,\bOmega,\bDelta,\nu)$, then $\Y$ has density function
$$
f_\Y(\y) = 2^{q} \mathrm{t}_{p}(\y\mid \bmu,\bSigma,\nu) {T}_q \left(\bDelta^\top\bSigma^{-1}(\y-\bmu) \sqrt{\frac{\nu+p}{\nu+{d}(\y)}} \mid \mathbf{0}, \bLambda, \nu+p\right), 
$$
where $\bSigma$ and $\bLambda$ are given in (\ref{eqn siglambda}) and $d(\y)=(\y-\bmu)^\top \bSigma^{-1} (\y-\bmu)$. 
\end{proposition}
\begin{proof} By definitions  \ref{def CFUSN} and \ref {def CFSUT}, we have that
\begin{equation} \label{eqn_repestocskewt}
\Y =\bmu+  U^{-1/2} (\bDelta |\X_0| + \X_1),
\end{equation}	
where $\X_0 \sim \mathbf{N}_q(\mathbf{0},\mathbf{I})$ and $\X_1 \sim \mathbf{N}_p(\mathbf{0},\bOmega)$ are independent. 
Let $\bV=U^{-1/2} \X_0$ and $\W=U^{-1/2}\X_1$. Then,
	\begin{equation*} 
\begin{pmatrix}
\bV\\
\W
\end{pmatrix} \sim \mathrm{t}_{q+p} 
\left[ \begin{pmatrix}\mathbf{0} \\ \mathbf{0} \end{pmatrix}, 
\begin{pmatrix}
\mathbf{I} & \mathbf{0}^\top\\
 \mathbf{0} & \bOmega
\end{pmatrix},\nu
\right].
\end{equation*}
Observe that  the distribution of $\Y$ is the same as  $\bmu+(\T|\bV>\mathbf{0})$, where $\T=\bDelta \bV+ \W$, and we have  
 	\begin{equation*} 
 \begin{pmatrix}
 \bV\\
\bmu+ \T
 \end{pmatrix} \sim \mathrm{t}_{q+p} 
 \left[ \begin{pmatrix}\mathbf{0} \\ \mathbf{\bmu} \end{pmatrix}, 
 \begin{pmatrix}
 \mathbf{I} & \bDelta^\top\\
 \bDelta & \bOmega+ \bDelta \bDelta^\top
 \end{pmatrix},\nu
 \right]. 
 \end{equation*}
Thus, the density of $\Y$ is $f_\Y(\y) =P(\bV>\mathbf{0})^{-1}P(\bV>\mathbf{0}|\bmu+\T=\y)f_{\bmu+\T}(\y)$.
Since $\bV \sim \mathrm{t}_q(\mathbf{0},\mathbf{I})$, $P(\bV>\mathbf{0})=2^{-q}$, and, by Lemma \ref{eqn t closed}, we know that $\bV|\bmu+\T=\y \sim \mathrm{t}_q(\bDelta^\top\bSigma^{-1}(\y-\bmu),(\nu+d(\y)/(\nu+p))\bLambda,\nu+p)$, which concludes the proof. 	
\end{proof}	

Similar to the SN case, by Definition \ref{def CFSUT} and Proposition \ref{prop affine transf SMSN}, it can be shown that affine transformations of ST distributions are still ST distributed, as stated in the following proposition. 
\begin{proposition}  \label{prop affine ST}
	Let  $\Ae$ be a $m \times p$ matrix,  $\be$ be a vector of length $m$, and $\Y \sim  ST_{p,q}(\bmu,\bOmega,\bDelta,\nu)$. Then,   $\Z= \Ae\Y+\be \sim  ST_{m,q}(\Ae\bmu  +\be,\Ae \bOmega \Ae',\Ae \bDelta,\nu)$. 
\end{proposition}

If $\Y=(\Y_1^\top,\Y_2^\top)^\top$, we have $\Y_1 \sim \textrm{ST}_{p_1, q}(\bmu_1,\bOmega_{11},\bDelta_{1},\nu)$. An analogous result can be shown for $\Y_2$.	

Equation \eqref{eqn_repestocskewt} also  implies the following result: 
\begin{proposition} \label{prop strep CFUST}
Let  $\Y \sim ST_{p,q}(\bmu,\bOmega,\bDelta,\nu)$ and $\bS =U^{-1/2} |\X_0|$. Then, $\Y$ admits the following hierarchical representation   
\begin{align*}
\Y| \bS=\bs, U=u  & \sim   \textrm{{N}}_{p}(\bmu + \bDelta \bs,u^{-1}\bOmega); \\
\bS|U=u  & \sim   \textrm{HN}_q(\mathbf{0},u^{-1} \mathbf{I}_q); \\
U & \sim Gamma(\nu/2,\nu/2), 
\end{align*}
where $\textrm{HN}_q(\mathbf{0},u^{-1} \mathbf{I}_q)$ denotes  the $q$-dimensional half-normal
distribution with location parameter $\mathbf{0}$ and scale matrix $u^{-1}\mathbf{I}_q$.	
\end{proposition}

It is straightforward to show that the density of $\bS|U=u $ in Proposition \ref{prop strep CFUST} is 
$$
f_{\bS|U}(\bs|u) = 2^q (2 \pi)^{-q/2} u^{q/2} \exp\left(-\frac{u}{2}\bs^\top \bs \right), \quad \bs > \mathbf{0}.  
$$ 

\begin{proposition}
	Let  $\Y \sim ST_{p,q}(\bmu,\bOmega,\bDelta,\nu)$. Then,
	\begin{equation} \label{eqn exp CFUST}
	\mathrm{E}(\Y)= \bmu + \sqrt{2/\pi}\kappa_1
	\bDelta \mathbf{1}_q \,\,\, \mbox{and} \,\,\, \mathrm{Var}(\Y) = \frac{\nu}{\nu-2}\left(\bSigma - {\frac{2}{\pi}} \bDelta \bDelta^\top\right) + a(\nu)\bDelta \,{\bf J}_q\,\bDelta^\top, 
	\end{equation}
where $\kappa_1 = (\nu/2)^{1/2} \Gamma\left(\frac{\nu-1}{2}\right)/\Gamma\left(\frac{\nu}{2}\right)$, $a(\nu)=\frac{2}{\pi}\left(\frac{\nu}{\nu-2} - \kappa_1^2\right)$  and ${\bf J}_q$ is a $q \times q$ matrix of ones.
\end{proposition}
To prove this result, notice that by \eqref{eqn_repestocskewt}, the distribution of $\Y|U=u$ is $SN_{p,q}(\bmu,u^{-1}\bOmega,u^{-1/2}\bDelta)$. Thus, by Proposition  \ref{prop Exp CFUSN}, we have that $\mathrm{E}(\Y|U=u)=  \bmu + u^{-1/2} \sqrt{{2}/{\pi}} \bDelta \mathbf{1}_q,$ which implies $\mathrm{E}(\Y)=\mathrm{E}[\mathrm{E}(\Y|U)]=  \bmu + \mathrm{E}(U^{-1/2}) \sqrt{{2}/{\pi}} \bDelta \mathbf{1}_q$. We have that $\mathrm{E}(U^{-1/2})= \kappa_1$,
which can be easily obtained since $U$ has a Gamma distribution, and the first part of the result follows. Similarly, the variance can be easily derived.

\section{The model}\label{sec:model}
\subsection{Definition}
The standard linear mixed model introduced by \cite{laird1982random} has been a widely used tool to model the correlation within-subjects often present in longitudinal data. In this work, we  present a flexible extension of this model. 

Suppose that  there are $N$ subjects, with the $i$th subject having $n_i$ observations, then a linear mixed model can be written as
\begin{equation} \label{def mod longitudinal}
\Y_i = \X_i\bbeta + \Z_i \be_i +  \bepsilon_i, \quad  i=1,\ldots,N, 
\end{equation}
We consider that 
\begin{equation} \label{eqn joint dist bi epsi}
\left( \begin{array}{c} \mathbf{b}_{i} \\\bepsilon_{i} \end{array}  \right)  \ind \mathrm{ST}_{q+n_i,r}\left[\left( \begin{array}{c} b\bDelta \mathbf{1}_r \\\mathbf{0}_{n_i \times 1}\end{array}  \right), \left( \begin{array}{cc} \D & \mathbf{0}_{q \times n_i} \\
\mathbf{0}_{n_i \times q} & \bOmega_i
 \end{array}\right),\left( \begin{array}{c} \bDelta \\\mathbf{0}_{n_i \times r} \end{array}  \right), \nu  \right],
\end{equation}
where  $b=-\sqrt{\nu/\pi} \Gamma\left(\frac{\nu-1}{2}\right)/\Gamma\left(\frac{\nu}{2}\right)$, and $\,\ind\,$
denotes independent  random vectors. 
This setup implies 
\begin{equation*} 
\be_i \iid  \mathrm{ST}_{q,r}(b\bDelta \mathbf{1}_r, \D,\bDelta,\nu), \quad \bepsilon_{i} \ind \mathrm{t}_{n_i}(\mathbf{0},\bOmega_i,\nu),  \quad i=1,\ldots,N,  
\end{equation*}
where $\iid $
denotes independent and identically distributed  random vectors. 
Thus, we have that $\textrm{E}[\be_i]=\mathbf{0}$, according to \eqref{eqn exp CFUST}. Moreover, from Proposition \ref{prop strep CFUST}, we have the following stochastic representation of the complete data model 
\begin{align}
\Y_i|\be_i,U_i=u_i  & \sim \mathrm{N}_{n_i}(\X_i\bbeta + \Z_i \be_i, u_i^{-1}\bOmega_i); \label{eqn first re 1}\\
\be_i|\bS_i=\s_i,U_i=u_i  & \sim  \mathrm{N}_{q}(\bDelta(b \mathbf{1}_r+ \s_i), u_i^{-1}\D); \label{eqn repest bi}\\
\bS_i|U_i=u_i  & \sim   \textrm{HN}_r(\mathbf{0},u_i^{-1} \mathbf{I}_r); \label{eqn repest si}  \\
U_i & \sim \mathrm{Gamma}(\nu/2,\nu/2),\label{eqn:repestui}  \quad i=1,\ldots,N.  
\end{align}

Observe that $\Y_i = \X_i\bbeta+ \Ae(\be_i^\top \,\, \bepsilon_i^\top)^\top$, with $\Ae=(\Z_i \,\, \mathbf{I}_{n_i})$. Using Proposition \ref{prop affine ST}, we have 
\begin{equation} \label{eqn dist complete data}
\Y_i  \sim \mathrm{ST}_{n_i,r}(\X_i\bbeta + b\Z_i \bDelta \mathbf{1}_r,\bPsi_i,\Z_i\bDelta,\nu),
\end{equation}
with 
$\bPsi_i=\Z_i\D \Z_i^\top+\bOmega_i.$
Hence, the marginal pdf of $\Y_i$ is
\begin{equation}\label{eqn:ydens}
\begin{aligned}
f(\y_i) ={} & 2^{r} \,\mathrm{t}_{n_i}\left(\y_i\mid \bmu_i,\bSigma_i,\nu\right)\times \\ &\mathrm{T}_r \left(\left.\bDelta^\top\Z_i^\top\bSigma_i^{-1}(\y_i-\bmu_i) \sqrt{\frac{\nu+n_i}{\nu+{d_i}(\y_i)}} \right| \mathbf{0}, \bLambda_i, \nu+n_i\right), \quad \ii, 
\end{aligned}
\end{equation}
where $\bmu_i = \X_i\bbeta + b\Z_i \bDelta \mathbf{1}_r$, $\bSigma_i = \bPsi_i + \Z_i\bDelta\bDelta^\top \Z_i^\top$, $\bLambda_i = \mathbf{I}_r - \bDelta^\top\Z_i^\top \bSigma_i^{-1} \Z_i\bDelta$, and $d_i(\y_i)=(\y_i-\bmu_i)^\top \bSigma_i^{-1} (\y_i-\bmu_i)$. 
Therefore, assuming that $\D = \D(\balpha)$, $\bDelta = \bDelta(\bdelta)$ and $\bOmega=\bOmega(\bphi)$ 
depend on unknown and reduced parameter vectors  $\balpha$, $\bdelta$ and $\bphi$, respectively, 
the log-likelihood function for $\btheta$ based on the observed sample
$\mathbf{y} = (\mathbf{y}^{\top}_1,\ldots,\mathbf{y}^{\top}_n)^{\top}$
is given by 
\begin{equation}\label{eq:loglik}
\ell(\btheta|\mathbf{y})=\sumas
\ell_i(\btheta|\mathbf{y})=\sumas
\log{f(\mathbf{y}_i|\btheta)},\end{equation}
where $\btheta=\left(\bbeta^{\top}, \bphi^{\top}, \balpha^{\top}, \bdelta^{\top}, \nu\right)^{\top}$.
Since the observed log-likelihood function involves complex expressions, it is very computationally expensive to work directly with $\ell(\btheta|\mathbf{y})$ to find the ML estimates of $\btheta$.
Hence, in the following subsection, we discuss the development of an EM-type algorithm \citep{Dempster77} for ML estimation.

\subsection{Maximum likelihood estimation}\label{sec:mlest}

From the hierarchical representation given in \eqref{eqn first re 1}--\eqref{eqn:repestui} and treating $\mathbf{b} = (\mathbf{b}^{\top}_1, \ldots, \mathbf{b}^{\top}_N)^{\top}$, $\mathbf{s} = (\mathbf{s}_1, \ldots, \mathbf{s}_N)^{\top}$ and $\mathbf{u} = (u_1, \ldots, u_N)^{\top}$ as hypothetical missing data, we propose to use the ECME algorithm \citep{Liu94} for parameter estimation.
Let the augmented data set be 
$\mathbf{y}_c = (\yp^{\top}, \be^{\top}, \mathbf{s}^{\top}, \mathbf{u}^{\top})^{\top}$, where
$\yp = (\yp^{\top}_1, \ldots, \yp^{\top}_N)^{\top}$, then the complete-data log-likelihood
function $\ell_c(\btheta|\mathbf{y}_c)=\sumas
\ell_i(\btheta|\mathbf{y}_c) $ is given by
\begin{eqnarray*}\label{logcompleta}
	\ell_c(\btheta|\mathbf{y}_c)&=&\sumas\left[-\frac{1}{2}
	\log{|\bOmega_i|}
	-\frac{u_i}{2}(\yp_i-\mathbf{X}_i\bbeta
	-\mathbf{Z}_i\mathbf{b}_i)^{\top}\bOmega_i^{-1}
	(\yp_i-\mathbf{X}_i\bbeta-\mathbf{Z}_i\mathbf{b}_i)\right.\nonumber\\
	&&\left.-\frac{1}{2}\log{|\D|}-
	\frac{u_i}{2}(\mathbf{b}_i-\bDelta(b \mathbf{1}_r+ \s_i))^{\top}\D^{-1}(\mathbf{b}_i- \bDelta(b \mathbf{1}_r+ \s_i))\right]+K(\nu|\mathbf{u},\mathbf{s}),
\end{eqnarray*}
where $K(\nu|\mathbf{u},\mathbf{s})$ is a function that depends on the parameter vector $\btheta$ only through $\nu$.

Given the current value $\btheta=\widehat{\btheta}^{(k)}$, the E-step of
an EM-type algorithm evaluates $\qhatk(\btheta) = \Ex\left\{\ell_c(\btheta|\mathbf{y}_c)\mid \widehat{\btheta}^{(k)}, \y\right\}=\sum^n_{i=1} \qhatk_i(\btheta)$,
where the expectation is taken with respect to the joint
conditional distribution of $\mathbf{b}$, $\mathbf{s}$, and $\mathbf{u}$, given $\mathbf{y}$ and $\widehat{\btheta}$. Therefore, we can write
$$\qhatk_i(\btheta)=\qhatk_{1i}(\bbeta,\bphi)+
\qhatk_{2i}(\balpha,\bdelta) + \qhatk_{3i}(\nu),$$ {where}
\begin{eqnarray*}
\qhatk_{1i}(\bbeta,\bphi) &=& -\frac{1}{2}\log|\bOmega_i| 
-\frac{\widehat{u}_i^{(k)}}{2}\left(\y_{i}-\X_i{\bbeta}\right)^{\top}{\bOmega}_i^{-1}\left(\y_{i}-\X_i{\bbeta}\right)\nonumber\\&&+
\left(\y_{i}-\X_i{\bbeta}\right)^{\top}{\bOmega}_i^{-1}\Z_i\widehat{\mathbf{ub}}_i^{(k)}
-\frac{1}{2}\textrm{tr}\left({\bOmega}_i^{-1}\Z_i\widehat{\mathbf{ub^2}}_i^{(k)}\Z^{\top}_i\right),\,\,\,\,
\\
\qhatk_{2i}(\balpha,\bdelta) &=& -\frac{1}{2}\log{\left| {\D} \right|}-\frac{1}{2}\textrm{tr}\left({\D}^{-1}\widehat{\mathbf{ub^2}}_i^{(k)}\right)+
b \,\widehat{\mathbf{ub}}_i^{(k)\top}{\D}^{-1}\bDelta \umr+\textrm{tr}\left({\D}^{-1}\bDelta\,\widehat{\mathbf{ubs}}_i^{(k)\top}\right) 
\\&& -
b \,\widehat{\mathbf{us}}_i^{(k)\top}{\D}^{-1}\bDelta \umr -\frac{1}{2}\textrm{tr}\left(\bDelta^\top {\D}^{-1}\bDelta \widehat{\mathbf{us^2}}_i^{(k)}\right)
- \frac{\widehat{u}_i^{(k)}}{2} b^2 \umr^\top\bDelta^\top {\D}^{-1}\bDelta \umr,
\end{eqnarray*}
and {$\qhatk_{3i} = \Ex\left\{K(\nu|\mathbf{u},\mathbf{s})\mid \widehat{\btheta}^{(k)}, \y \right\}$}, with $\textrm{tr}(\A)$ indicating the trace of matrix $\A$, and 
$\widehat{u}_i^{(k)} = \Ex\{U_i\mid \widehat{\btheta}^{(k)},\y_i\}$,
$\widehat{\mathbf{u}\be}_i^{(k)} = \Ex\{U_i \be_i\mid \widehat{\btheta}^{(k)},\y_i\}$,
$\widehat{\mathbf{ub^2}}_i^{(k)} = \Ex\{U_i \be_i\be_i^\top\mid \widehat{\btheta}^{(k)},\y_i\}$,
$\widehat{\mathbf{ubs}}_i^{(k)} = \Ex\{U_i \be_i\bS_i^\top\mid \widehat{\btheta}^{(k)},\y_i\}$,
$\widehat{\mathbf{us}}_i^{(k)} = \Ex\{U_i \bS_i\mid \widehat{\btheta}^{(k)},\y_i\}$,
and $\widehat{\mathbf{us^2}}_i^{(k)} = \Ex\{U_i \bS_i\bS_i^\top\mid \widehat{\btheta}^{(k)},\y_i\}$,
$\ii$.

From the representation given in \eqref{eqn first re 1}--\eqref{eqn:repestui}, using properties from conditional expectation and after some algebra, omitting the supra-index $(k)$, the expressions above can be written as:

\begin{equation*}
    \begin{array}{l}
    \widehat{u}_i = \dfrac{\widehat{\nu}+n_i}{\widehat{\nu}+\widehat{d}_i(\y_i)} \dfrac{\mathrm{T}_r\left(\widehat{\bq}_i(\y_i)\sqrt{({\widehat{\nu}+n_i+2})/({\widehat{\nu}+\widehat{d}_i(\y_i)})} \left| \mathbf{0},\, \widehat{\bLambda}_i,\,\widehat{\nu}+n_i+2\right.\right)}{\mathrm{T}_r\left(\widehat{\bq}_i(\y_i)\sqrt{({\widehat{\nu}+n_i})/({\widehat{\nu}+\widehat{d}_i(\y_i)})} \left| \mathbf{0},\, \widehat{\bLambda}_i,\,\widehat{\nu}+n_i\right.\right)},\\
    \widehat{\mathbf{us}}_i = \widehat{u}_i\,\Ex\{W_i\mid \widehat{\btheta}, \y_i\}, \quad
    \widehat{\mathbf{us^2}}_i = \widehat{u}_i\,\Ex\{W_iW_i^\top\mid \widehat{\btheta}, \y_i\}, \\ 
  \widehat{\mathbf{u}\be}_i = \widehat{\br}_i\ui + \widehat{\M}_i\widehat{\D}^{-1}\widehat{\bDelta}\, \widehat{\mathbf{us}}_i, 
  \quad
  \widehat{\mathbf{ubs}}_i = \widehat{\br}_i\,\widehat{\mathbf{us}}^\top_i + \widehat{\M}_i\widehat{\D}^{-1}\widehat{\bDelta}\, \widehat{\mathbf{us^2}}_i,
  \\
   \widehat{\mathbf{ub^2}}_i = \widehat{\M}_i + \widehat{\mathbf{ub}}_i \,\widehat{\br}_i^\top+ \widehat{\mathbf{ubs}}_i \,\widehat{\bDelta}^\top
   \widehat{\D}^{-1}\widehat{\M}_i,
    \end{array}
\end{equation*}
where $\bq_i(\y_i) = \bDelta^\top\Z_i^\top\bSigma_i^{-1}(\y_i-\bmu_i)$, $\br_i = b \bDelta\umr + \M_i \Z_i^\top \bOmega_i^{-1}(\y_i-\bmu_i)$, $\M_i = \left( \D^{-1} +\Z_i^\top \bOmega_i^{-1}\Z_i\right)^{-1},$
and $d_i(\y_i), \bSigma_i, \bmu_i$ and $\bLambda_i$ are as given in \eqref{eqn:ydens}. Moreover,
\begin{equation}
W_i\mid {\btheta}, \y_i \,\sim\, \mathrm{TT}_r\left({\bq}_i(\y_i), \, \frac{{{\nu}+{d}_i(\y_i)}}{{{\nu}+n_i+2}}{\bLambda}_i,\,{\nu}+n_i+2 ; \,\mathbb{R}^+\right),
\end{equation}
with $\mathrm{TT}_r(\bmu,\bSigma,\nu;\mathbb{A})$ denoting a Student-t distribution truncated on $\mathbb{A}$, and its moments can be computed using the \textsf{R} package \textbf{MomTrunc} \citep{galarza2020momtrunc}.

To maximize $\qhatk(\btheta)$ with respect to $\btheta$, the ECME algorithm performs a conditional maximization (CM) step that conditionally maximizes $\qhatk(\btheta)$, obtaining a new estimate  $\widehat{\btheta}^{(k+1)}$, as follows: 

\noindent{\bf 1.} $\widehat{\bbeta}^{(k)}$, $\widehat{\bphi}^{(k)}$, $\widehat{\D}^{(k)}$, and $\widehat{\bDelta}^{(k)}$ are updated using the following expressions:
\begin{eqnarray*}
\widehat{\bbeta}^{(k+1)}&=&\left(\sumas\uki
\X_i^{\top}\widehat{\bOmega}^{{-1}}_i(\widehat{\bphi}^{(k)})\X_i\right)^{-1}\sumas
\X_i^{\top}\widehat{\bOmega}^{{-1}}_i(\widehat{\bphi}^{(k)})\left(\uki\yp_i-\Z_i
\widehat{\mathbf{u}\be}_i^{(k)}\right),\\
\widehat{\bphi}^{(k+1)}&=&\underaccent{\bphi}{\textrm{argmax}} \left\{\sumas\qhatk_{1i}(\widehat{\bbeta}^{(k+1)},\bphi)\right\},\\
\widehat{\D}^{(k+1)}&=& \frac{1}{N}\sumas \frac{\widehat{K}_{1i}^{(k)}+\widehat{K}_{1i}^{(k)\top}}{2},\\
\widehat{\bDelta}^{(k+1)}&=&
\sumas \left(\widehat{b}^{(k)} \widehat{\mathbf{ub}}_i^{(k)}\umr^\top + \widehat{\mathbf{ubs}}_i^{(k)}\right) \left(\sumas  \widehat{\mathbf{us^2}}_i^{(k)} + \widehat{u}_i^{(k)}\widehat{b}^{2(k)} \umr \umr^\top +\widehat{b}^{(k)}(\umr\widehat{\mathbf{us}}_i^{(k)\top} + \widehat{\mathbf{us}}_i^{(k)}\umr^\top) \right)^{-1},\\
\end{eqnarray*}
where $ \widehat{b}^{(k)} = b( \widehat{\nu}^{(k)})$ \, and \, $\widehat{K}_{1i}^{(k)} =	\widehat{\mathbf{ub^2}}_i^{(k)}-
2 \widehat{\bDelta}^{(k)}	\widehat{\mathbf{ubs}}_i^{(k)\top} +
2 \widehat{b}^{(k)} \widehat{\bDelta}^{(k)} \umr \left(\widehat{\mathbf{us}}_i^{(k)\top}\widehat{\bDelta}^{{(k)\top}}- \widehat{\mathbf{ub}}_i^{(k)\top} \right) +\\
\widehat{\bDelta}^{(k)}\widehat{\mathbf{us^2}}_i^{(k)}\widehat{\bDelta}^{{(k)\top}} +
\widehat{u}_i^{(k)} \widehat{b}^{2(k)} \widehat{\bDelta}^{(k)}\umr \umr^\top \widehat{\bDelta}^{{(k)\top}}$.

\noindent{\bf 2.} $\widehat{\nu}^{(k)}$ is updated by
optimizing the constrained actual marginal log-likelihood
function $
\ell\left(\widehat{\btheta}^{*(k+1)}, \nu\left|\mathbf{y}\right.\right)$ given in \eqref{eq:loglik}, where {$\btheta^{*} = (\bbeta^{\top},  \bphi^{\top}, \balpha^{\top}, \bdelta^{\top})^{\top}$.} 
For computational efficiency, this optimization is restricted to the case where $ \{\nu\in\mathbb{N}|\nu>1\}$.

The update of $\widehat{\bphi}^{(k)}$ depends on the specific structure considered for the scale matrix of the within-subject error term. If we restrict to the conditionally uncorrelated case, where $\bOmega_i = \sigma_e^2 \textbf{I}_{n_i}$, then $\bphi = \sigma_e^2$, and the update can be written as
\begin{eqnarray*}
\widehat{\sigma^2_e}^{{(k+1)}}=\frac{1}{\sumas n_i}\,\,\sumas
\left[\uki\left(\yp_i-\mathbf{X}_i\widehat{\bbeta}^{(k+1)}\right)^{\top}
\left(\yp_i-\mathbf{X}_i\widehat{\bbeta}^{(k+1)}\right) -\right.\\\left. 
2\left(\yp_i-\mathbf{X}_i\widehat{\bbeta}^{(k+1)}\right)^{\top}
\Z_i\widehat{\mathbf{u}\be}_i^{(k)} +
{\textrm{tr}}\left(\z_i\widehat{\mathbf{ub^2}}_i^{(k)} \z^{\top}_i\right)\right].
\end{eqnarray*}

The algorithm is iterated until a predefined criterion is reached, such as when $ \left|\dfrac{\ell(\widehat{\btheta}^{(k+1)}\mid \y)}{\ell(\widehat{\btheta}^{(k)}\mid \y)}-1\right|$ becomes smaller than a predefined value.

\subsection{Estimation of random effects and standard errors}\label{subsec:error}

To obtain an estimate of the random effects, we compute the minimum mean-squared error (MSE) estimator of $\mathbf{b}_i$,
that is given by the conditional mean of $\mathbf{b}_i$ given $\Y_i=\y_i$, as follows:
	\begin{eqnarray}
	\widehat{\be}_i(\btheta)&=&\Ex\{\mathbf{b}_i|\y_i,\btheta\}=\Ex_{U_i}\{\Ex_{\bS_i}\{\Ex_{\be_i}\{\mathbf{b}_i|\bs_i,u_i,\y_i,\btheta\}|u_i,\y_i,\btheta\}|\y_i,\btheta\}\nonumber\\\label{eq:estbi}
	 &=& b \bDelta\umr + \M_i \Z_i^\top \bOmega_i^{-1}(\y_i-\bmu_i) + 
	 \M_i \D^{-1} \bDelta\, \Ex\{W^*_i|{\btheta}, \y_i\},
	\end{eqnarray}
where
\begin{equation*}
W_i^*\mid {\btheta}, \y_i \,\sim\, \mathrm{TT}_r\left({\bq}_i(\y_i), \, \frac{{{\nu}+{d}_i(\y_i)}}{{{\nu}+n_i}}{\bLambda}_i,\,{\nu}+n_i ; \,\mathbb{R}^+\right).
\end{equation*}
In practice, the empirical Bayes estimator of $\mathbf{b}_i$ can be obtained by substituting the ML
estimate $\widehat{\btheta}$ into \eqref{eq:estbi}, that is, 
$\widehat{\mathbf{b}}_i = \widehat{\mathbf{b}}_i(\widehat{\btheta})$.

In addition, to obtain standard errors estimates, following \cite{matos2018multivariate} and assuming some general regularity conditions, we compute the empirical information matrix using the complete-data gradient vector with respect to $\btheta^*=\btheta\setminus\nu$, based on results of \cite{louis1982finding}. 
Evaluated at the EM estimate $\widehat{\btheta}$, the empirical information matrix is given by
\begin{equation}
    \textbf{I}_e(\widehat{\btheta}^*|\y) = \sumas \widehat{\bs}_i \,\widehat{\bs}_i^\top,  
\end{equation}
where $\widehat{\bs}_i = \bs \left(\y_i|\widehat{\btheta}^*\right)$, in which
$\bs (\y_i|{\btheta}^*)=\dfrac{\partial \log{f(\y_i|\btheta^*,\nu)}}{\partial \btheta^*}=
\Ex \left\{\dfrac{\partial \ell_i(\btheta^*,\nu|\y_{ic})}{\partial \btheta^*}| \y_i,\btheta^*,\nu \right\}$ and
$\ell_i(\btheta|\y_{ic})$ is the complete data log-likelihood from the $i$th observation vector $\y_i$, $\ii$.

Restricting to the conditionally uncorrelated case (where $\bphi = \sigma_e^2$), \ $\widehat{\bs}_i$ 
is a vector of dimension $p+1+q(q+1)/2+q\,r$ with the following components:
\begin{eqnarray*}
\widehat{\bs}_{i,\beta} &=&  \X_i^\top\widehat{\bOmega}^{{-1}}_i\left( \widehat{u}_i(\y_i - \X_i\widehat{\bbeta}) - \Z_i \widehat{\mathbf{u}\be}_i\right),\\
\widehat{\bs}_{i,\sigma^2_e}&=& -\frac{n_i}{2\widehat{\sigma}_e^2} + \frac{1}{2\widehat{\sigma}_e^4}\left[\widehat{u}_i\left(\yp_i-\mathbf{X}_i\widehat{\bbeta}\right)^{\top}
\left(\yp_i-\mathbf{X}_i\widehat{\bbeta}\right) - 
2\left(\yp_i-\mathbf{X}_i\widehat{\bbeta}\right)^{\top}
\Z_i\widehat{\mathbf{u}\be}_i +
{\textrm{tr}}\left(\z_i\widehat{\mathbf{ub^2}}_i \z^{\top}_i\right)\right], \\ 
\widehat{\bs}_{i,\alpha} &=& \text{upper.tri} \left( -\frac{1}{2}\widehat{\D}^{-1}+\frac{1}{2}\widehat{\D}^{-1}\frac{\widehat{K}_{1i}+\widehat{K}_{1i}^{\top}}{2} \widehat{\D}^{-1} \right),\\
\widehat{\bs}_{i,\delta}&=& \text{vec}\left(\widehat{\D}^{-1} \left(\widehat{b}\, \widehat{\mathbf{ub}}_i\umr^\top + \widehat{\mathbf{ubs}}_i\right) -\widehat{\D}^{-1} \widehat{\bDelta}\left(\widehat{\mathbf{us^2}}_i + \widehat{u}_i\widehat{b}^{2} \umr \umr^\top +\widehat{b}(\umr\widehat{\mathbf{us}}_i^{\top} + \widehat{\mathbf{us}}_i\umr^\top)\right) \right),
\end{eqnarray*}
where the expectations $\widehat{u}_i, \widehat{\mathbf{u}\be}_i, \widehat{\mathbf{ub^2}}_i, \widehat{\mathbf{ubs}}_i, \widehat{\mathbf{us}}_i$ and $\widehat{\mathbf{us^2}}_i$ are computed as in 
Section \ref{sec:mlest}, and the operator upper.tri$(\A)$ extracts and vectorizes the elements of the upper triangular part of matrix $\A$ (including its diagonal), and vec$(A)$ stacks the columns of matrix $\A$.

\section{Simulation studies}\label{sec:sim}

To exemplify the flexibility of the proposed model and investigate its empirical properties, this section presents an illustrative example and two simulation studies.

\subsection{Illustrative study}
This study illustrates the flexibility of the ST distribution discussed in this work by generating only one sample from 
$$\be_i = (b_{0i},b_{1i})^\top \iid \mathrm{ST}_{2,2}\left(
 b \bDelta \mathbf{1}_r, 
\left(\begin{array}{rr}
      0.5&-0.2\\ -0.2&0.5
\end{array}\right), 
\bDelta, 10
\right), \quad i=1,\hdots, 200,$$
with $b \approx -0.949$, and the following scenarios are considered for $\bDelta$:
$$(a)\, \left(\begin{array}{rr}
      0.6&1.5\\-1.0&3.0 
\end{array}\right);
\,\,\,
(b)\, \left(\begin{array}{rr}
      1.7&0.7\\3.9&-0.8 
\end{array}\right);
\,\,\,
(c)\, \left(\begin{array}{rr}
      2.0&0.0\\0.0&-2.0 
\end{array}\right);
\,\,\, \text{ and } \,\,\,
(d)\, \left(\begin{array}{rr}
      2.0&0.0\\-2.0&0.0 
\end{array}\right).
$$
It is noteworthy that scenario $(c)$ is the particular case of the SDB-ST distribution proposed by \cite{Sahu_Dey_Marcia}, and scenario $(d)$ is equivalent to the case of $r=1$ considered in \cite{schumacher2021scale}, for example.

\begin{figure}[ht]
    \centering
    \includegraphics[width=1\textwidth]{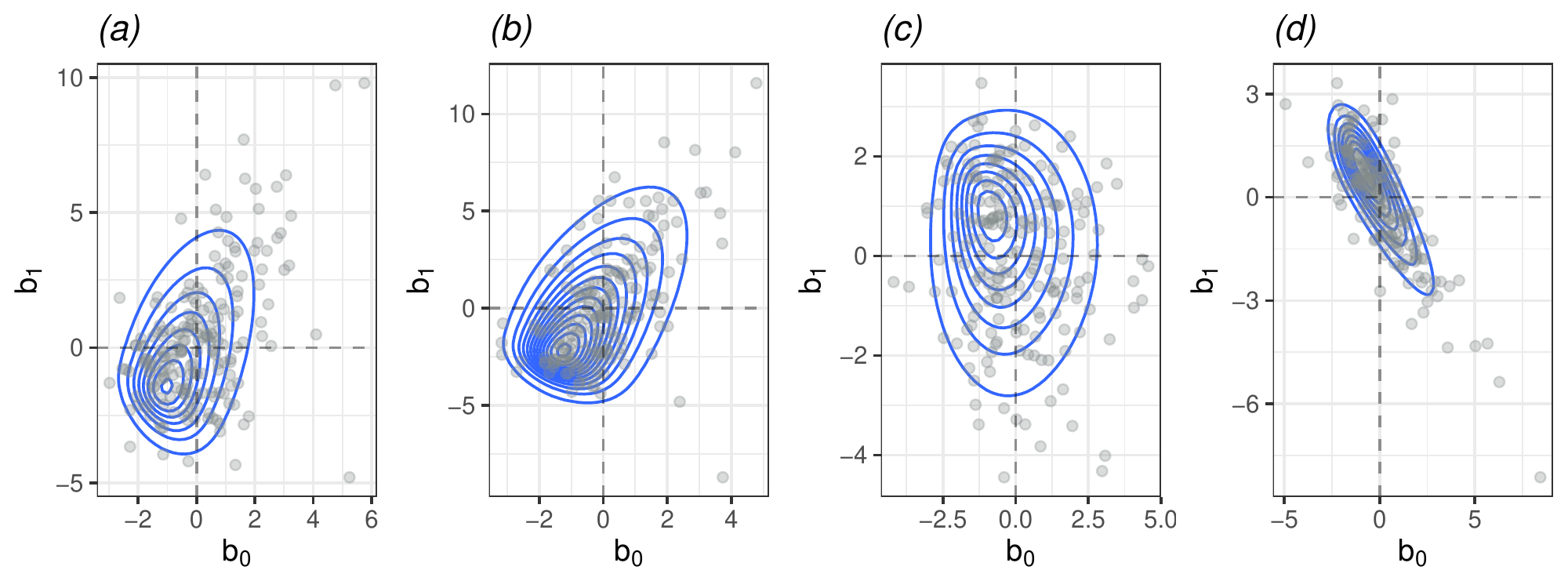}
    \caption{Illustrative study. Contour plots of the theoretical distribution of the random effects (blue curves), and simulated random effects (gray dots), for each scenario.}
    \label{fig:sim0}
\end{figure}

Figure \ref{fig:sim0} presents contour plots of the theoretical distribution of the random effects for each scenario, superimposed with generated values of the random effects. 
The great flexibility of the more general formulation of the ST distribution can be observed since its density function assumes various shapes.

\subsection{First study} 
In order to evaluate the empirical properties of the proposed model, we generated and estimated 500 Monte Carlo samples from the model
\begin{equation*}
    \Y_i = (\beta_0 + b_{0i})\mathbf{1}_{n} + (\beta_1 + b_{1i}){\bf x} + \bepsilon_i, \quad  i=1,\ldots,N, 
\end{equation*}
where $\beta_0=1$, $\beta_1=3$, $n = 5$, ${\bf x} = (-1, -0.5, 0, 0.5, 1)^\top$ and $N$ taking values $100, 200, 400$, and $600$, with $\bepsilon_{i} \ind \mathrm{t}_{5}(\mathbf{0},0.25 \mathbf{I}_{5},5)$ and 
$$\be_i = (b_{0i},b_{1i})^\top \iid \mathrm{ST}_{2,2}\left(
\left(\begin{array}{c}
       -1.993\\-1.898
\end{array}\right), 
\left(\begin{array}{rr}
      0.5&-0.2\\ -0.2&0.5
\end{array}\right), 
\left(\begin{array}{rr}
      0.6&-1.0\\1.5&3.0 
\end{array}\right), 5
\right).$$

The ML estimates and their associated SEs were recorded. For the sake of comparison, a numerical approximation of the Hessian matrix of the marginal likelihood function in \eqref{eq:loglik} using the function \verb|Hessian()| from the \textsf{R} package \textbf{numDeriv} \citep{rpackage_numderiv} was also recorded. In order to evaluate the consistency of the standard error estimation described in Subsection \ref{subsec:error}, we computed for each sample size the standard error of the ML estimates obtained from the 500 Monte Carlo samples (MC-SD), the mean of the standard error obtained as the diagonal of the inverse of the negative numerical Hessian (SE-N), and the mean of the standard error obtained using Louis' method, as presented in Subsection \ref{subsec:error} (SE-L). 

\begin{table}[ht]
\caption{Simulation study 1. Results based on 500 Monte Carlo samples for different numbers of subjects ($N$). MC-AV and MC-SD refer to the mean and standard deviation of the
estimates, respectively. SE-L and SE-N denote the average of standard errors obtained using Louis' and numerical methods.}\label{tabSim1}
\small \centering
\begin{tabular}{@{}crrrrrrrrrrr@{}}
\toprule
 & $\beta_0$ & $\beta_1$ & $\sigma^2$ & $\D_{11}$ & $\D_{12}$ & $\D_{22}$ & $\bDelta_{11}$ & $\bDelta_{21}$ & $\bDelta_{12}$ & $\bDelta_{22}$ & $\nu$ \\ \midrule
True & 1.000 & 3.000 & 0.250 & 0.500 & -0.200 & 0.500 & 0.600 & -1.000 & 1.500 & 3.000 & 5.000 \\
\midrule \multicolumn{12}{c}{$N$ = 200} \\\midrule
MC-AV & 0.992 & 2.998 & 0.252 & 0.470 & -0.200 & 0.555 & 0.627 & -0.918 & 1.503 & 2.969 & 5.398 \\
MC-SD & 0.078 & 0.114 & 0.020 & 0.135 & 0.129 & 0.246 & 0.258 & 0.292 & 0.169 & 0.221 & 0.854 \\
SE-L & 0.101	&0.164	&0.019	&0.155	&0.326	&0.475	&0.301	&0.473&	0.169&	0.257\\
SE-N & 0.100 & 0.162 & 0.019 &0.149&0.142&0.516&0.298&0.463&0.171&0.262&  \\
\midrule\multicolumn{12}{c}{$N$ = 600} \\\midrule
MC-AV & 0.995 & 2.999 & 0.253 & 0.491 & -0.198 & 0.511 & 0.612 & -0.989 & 1.513 & 3.006 & 5.246 \\
MC-SD & 0.041 & 0.055 & 0.012 & 0.076 & 0.065 & 0.100 & 0.126 & 0.126 & 0.077 & 0.109 & 0.492 \\

SE-L & 0.058	&0.095	&0.011	&0.087&	0.176	&0.275	&0.162	&0.250	&0.094	&0.146\\
SE-N & 0.058 & 0.094 & 0.011 &0.087	&0.079&0.288&0.172&0.262&0.097&0.146&  \\ \bottomrule
\end{tabular}
\end{table}

Table \ref{tabSim1} presents the standard error estimates for $N=200$ and $N=600$, in addition to the average of the ML estimates, denoted by MC-AV, and the parameter values used in data generation (True). 
In general, the point estimation seems close to the parameter value, and standard error estimates obtained using Louis' method seem reasonable for $\widehat{\bbeta}$ but can be misleading for parameters related to the random effects.
The estimates obtained from the numerical Hessian do not seem to improve the general accuracy, and it is  worth noting that they are based respectively on $458$ and $493$ samples for $N=200$ and $N=600$, as some samples resulted in numerical errors. Furthermore, the computational cost from the numerical method is much higher than the one from Louis' method.

\begin{figure}[ht]\centering
    \centering
    \includegraphics[width=.9\textwidth]{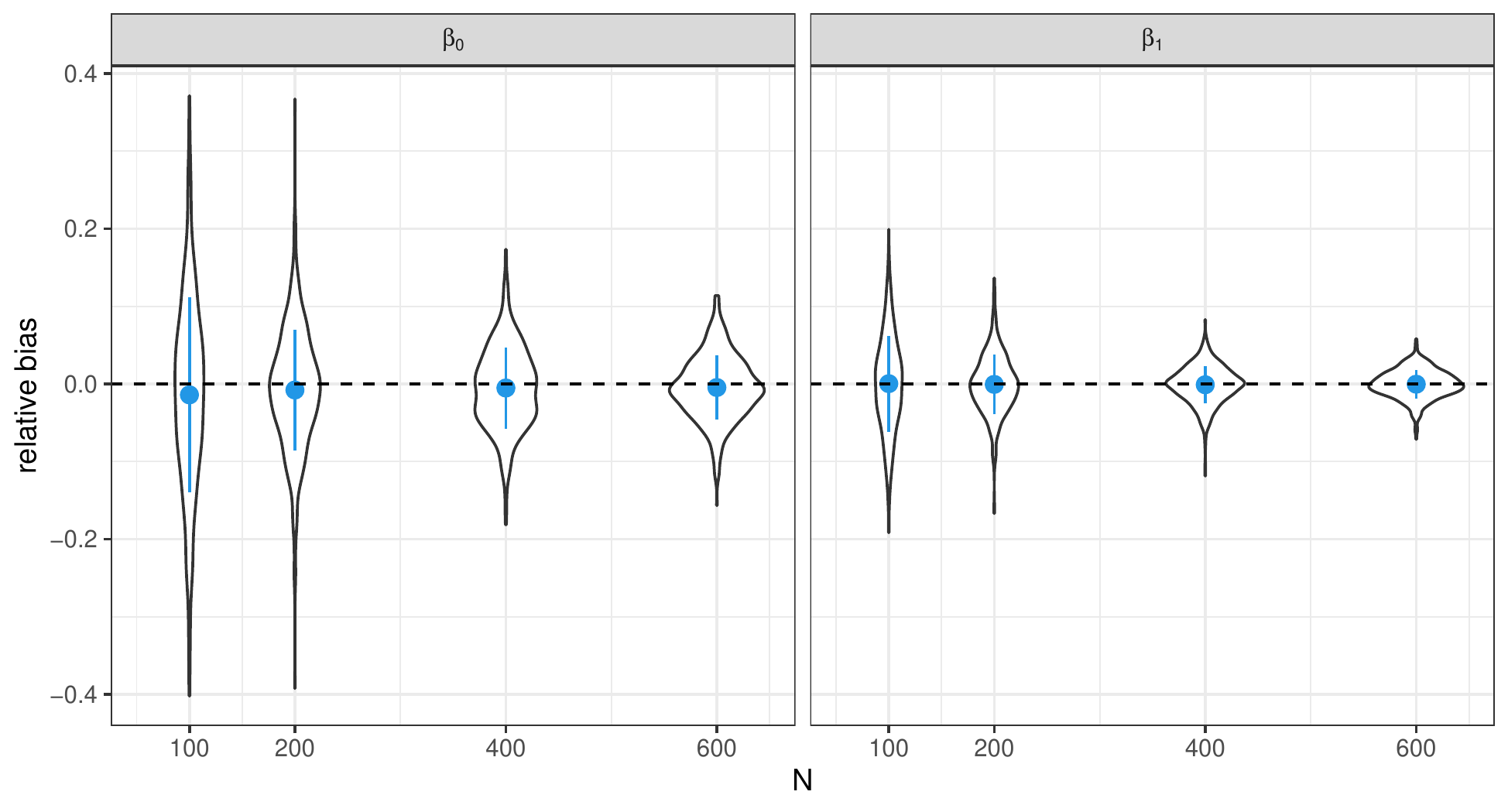}
    \caption{Simulation study 1. Violin plot of the relative bias of $\hat{\beta}_0$, and $\hat{\beta}_1$ for different numbers of subjects ($N$). Blue dots and lines indicate the mean and $\pm 1$ SD, respectively.}
    \label{fig:sim2assint}
\end{figure}

The consistency of the ML estimator of $\widehat{\bbeta}$ is illustrated in Figure \ref{fig:sim2assint}, which presents the mean relative bias and $\pm 1$ SD by the number of subjects ($N$). It can be seen that as the number of subjects increases, the bias (when it exists) draws closer to zero, and its SD decreases.

\subsection{Second study}
Aiming to evaluate the effect of different initial values, we 
considered 500 Monte Carlo samples from a model similar to the one estimated in the real data application presented in Section \ref{sec:applic}:
\begin{equation*}
    \Y_i = (\beta_0 + b_{0i})\mathbf{1}_{n} + (\beta_1 + b_{1i}){\bf x} + \beta_2\,{\bf x}^2 + \bepsilon_i, \quad  i=1,\ldots,200, 
\end{equation*}
where $\beta_0=2.7$, $\beta_1=-1$, $\beta_2=6.8$, $n = 5$, ${\bf x} = (-0.3, -0.15, 0, 0.15, 0.3)^\top$, and ${\bf x}^2 = ({x}_1^2,\hdots, {x}_5^2)^\top$ with $\bepsilon_{i} \ind \mathrm{t}_{5}(\mathbf{0},0.21 \mathbf{I}_{5},5)$ and 
$$\be_i = (b_{0i},b_{1i})^\top \iid \mathrm{ST}_{2,2}\left(
\left(\begin{array}{c}
       -2.278\\-2.942
\end{array}\right), 
\left(\begin{array}{rr}
      0.1&-0.1\\ -0.1&0.5
\end{array}\right), 
\left(\begin{array}{rr}
      1.7&0.7\\3.9&-0.8 
\end{array}\right), 5
\right).$$

For all scenarios, the estimation procedure was initialized with $\nu = 10$. The remaining parameters were initialized considering several approaches, as follows: 
\begin{itemize}\setlength{\itemsep}{0em}
    \item[(a)] The true parameter values plus a small normally generated error;
    \item[(b)] For $\bbeta, \sigma^2$ and $\D$, the estimated values from the normal LMM (obtained through the \texttt{lme()} function from \textbf{nlme} package in \textsf{R}); for $\bDelta$, the value that maximizes the marginal log-likelihood function given in \eqref{eq:loglik} on a small grid of $\bDelta$ and for other parameters fixed. 
    \item[(c)] Using the procedure described in (b) to estimate an SN-LMM (using up to 100 iterations), and then using the SN-LMM estimates as initial values for the ST-LMM;
    \item[(d)]  Using the procedure described in (b) to obtain initial values for $\bbeta$ and $\sigma^2$, and the procedure described in (c) for $\D$ and $\bDelta$;
    \item[(e)] Fitting the model using (b), (c), and (d), and then using as final estimate the fit that presents the highest likelihood value.
\end{itemize}
The proposal in (a) is impractical in real applications as the true parameter values are unavailable, but this scenario was considered for comparison purposes. 

\begin{figure}[ht]\centering
    \centering
    \includegraphics[width=1\textwidth]{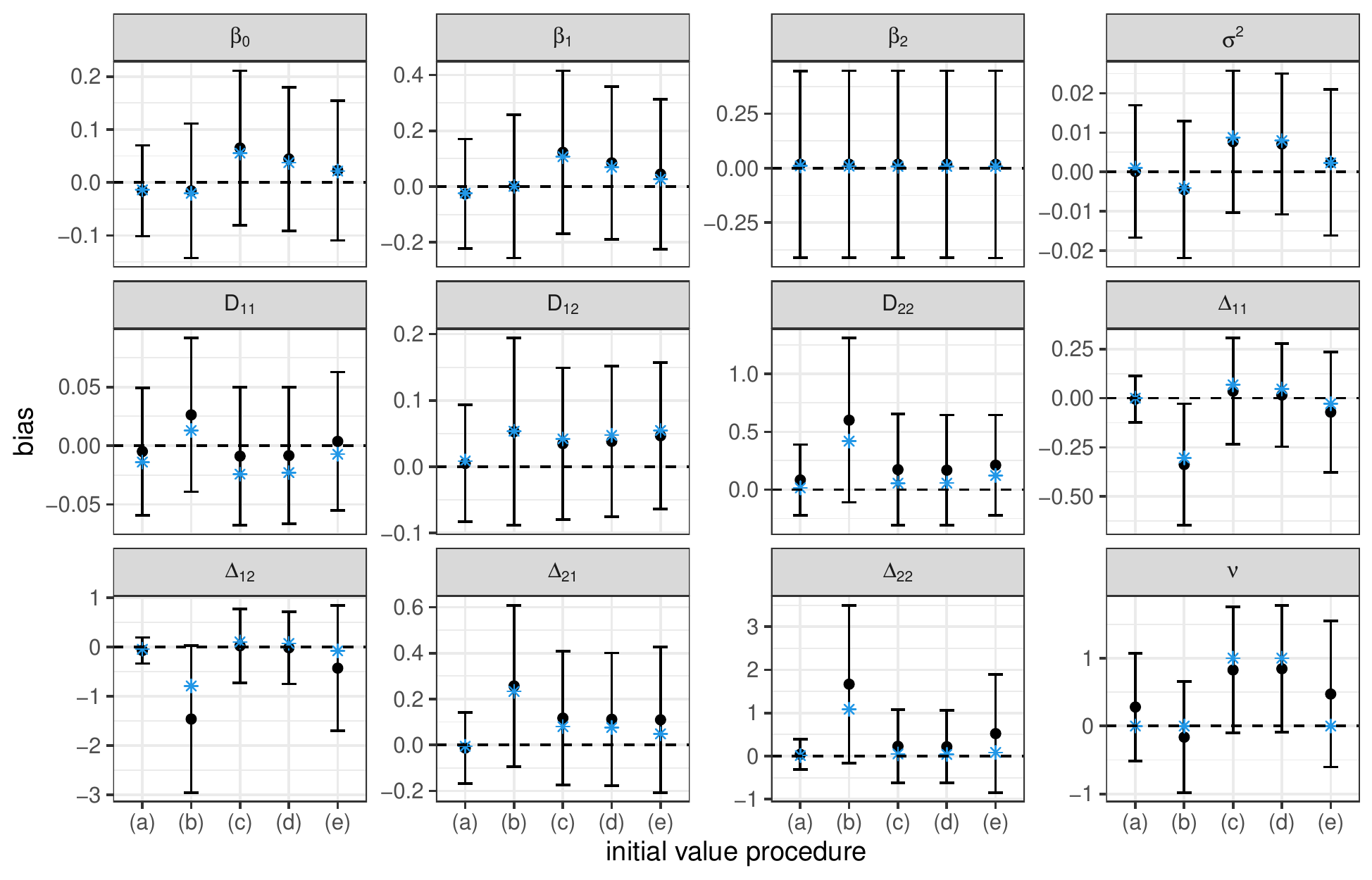}
    \caption{Simulation study 2. Mean bias $\pm 1$ SD of estimates for different initial value procedures. The blue asterisk indicates the median bias.}
    \label{fig:sim3initial}
\end{figure}

Figure \ref{fig:sim3initial} presents some summary information regarding the bias for all parameters and all initial value schemes. It can be seen that the method described in (b) is closer to (a) for estimating $\bbeta$, but is quite poor in estimating $\D$ and $\bDelta$. On the other hand, (c) and (d) are closer to (a) for estimating $\D$ and $\bDelta$, but seem to be biased to estimate $\bbeta$. Finally, (e) seems to perform satisfactorily to estimate all parameters, and therefore we consider this procedure in the practical application presented in the next section. 
Furthermore, it is worth mentioning that from the 500 Monte Carlo samples, method (e) selected (b) in 191 samples, (c) in 51 samples, and (d) in 258 samples.

\section{Application: schizophrenia data}\label{sec:applic}
Schizophrenia is a severe psychiatric disorder, and the equivalence of a new antipsychotic drug in comparison to a standard drug for this disorder was studied by \cite{lapierre1990controlled} using a double-blinded clinical trial  with randomization among four treatments: three doses (low, medium, and high) of a new therapy (NT) and a standard therapy (ST), for 245 patients with acute schizophrenia. 
The study was conducted at 13 clinical centers, and the primary response variable was assessed at the baseline (week 0) and weeks 1, 2, 3, 4, and 6 of treatment using the Brief Psychiatric Rating Scale (BPRS), which measures the extent of 18 features and rates each one on a seven-point scale, in which a higher number reflects a worse evaluation. The total BPRS score is the sum of the scores on the 18 items.

\begin{figure}\centering
    \centering
    \includegraphics[width=.9\textwidth]{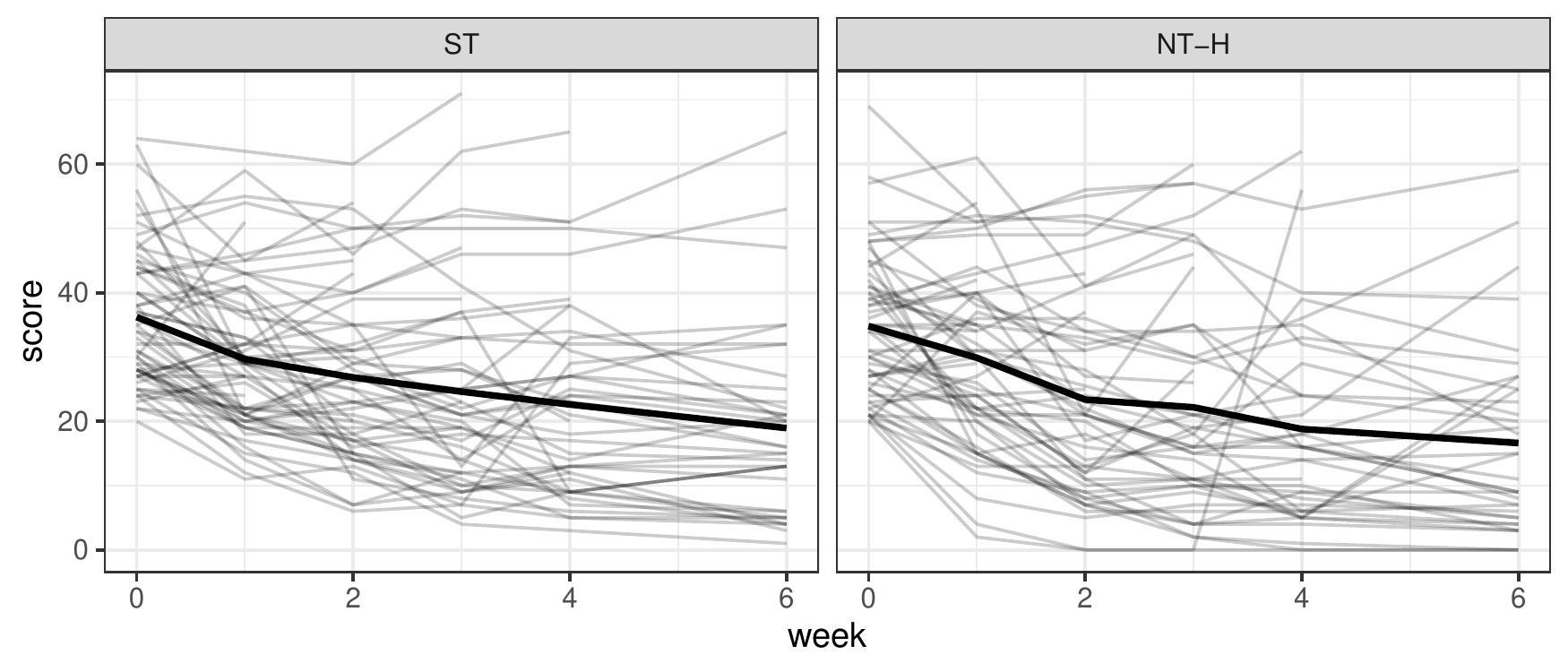}
    \caption{Trajectories of schizophrenia levels for the data. The thicker solid line indicates
the mean profile in the treatment.} 
\label{fig:schizo}

\end{figure}

Initial studies prior to this double-blinded study suggested that the experimental drug had equivalent antipsychotic activity, with fewer side effects. For the sake of simplicity, we will consider only the ST and the high dose of the NT, where each patient has at least two BPRS scores (118 patients), but an extension for modeling all treatments is direct. Individual BPRS trajectories evolved over up to six visits, as well as their mean profiles, are shown in Figure \ref{fig:schizo}, where it can be seen that several patients did not complete the study follow-up and that the mean profiles exhibit an apparent non-linear decline.
Furthermore, previous studies, such as \cite{ho2010robust} and \cite{schumacher2021scale}, showed that both subject-specific intercepts and slopes are positively skewed and that the data present heavy tails, indicating the need for a robust model that accommodates the random effect skewness.

In this section, we revisit this data set to allow for a more flexible skewness structure.
Based on the trajectories presented in Figure \ref{fig:schizo} and aiming to evaluate the treatment effect over time, we propose to fit the model 
$$ \Y_{i} = (\beta_0+ b_{0i})\mathbf{1}_{n_i} + (\beta_1+b_{1i}) {\bf x}_i+ \beta_2 {\bf x}_i^2 +\beta_3\,\textbf{NT}_i+\beta_4\,{\bf x}_i\,\textbf{NT}_i+ \bepsilon_{i}, \,\, i=1,\hdots , 118,$$
where $\Y_{i}$ is the $i$th participant total BRPS score vector divided by 10, $\mathbf{1}_{n_i}$ is the all-ones vector of length $n_i$, ${\bf x}_i=(x_{i1},\hdots,x_{in_i})^\top$, with $x_{ij}$ taken as (time - 3)/10 and time being measured in weeks from the baseline, ${\bf x}^2_i = (x_{i1}^2,\hdots,x_{in_i}^2)^\top$, and $\textbf{NT}_i$ is an all-ones vector if the $i$th subject received the new therapy and an all-zeros vector otherwise (an indicator vector of receiving NT).

\begin{table}[ht]
\centering
\caption{Selection criteria for fitting the SN-LMM and ST-LMM to the schizophrenia data set.}\label{tab:schizo1}
\begin{tabular}{@{}c|rrr|rrr@{}}
\toprule
criterion & \multicolumn{1}{c}{SN (r=1)} & \multicolumn{1}{c}{SN (SDB)} & \multicolumn{1}{c|}{SN (r=2)} & \multicolumn{1}{c}{ST (r=1)} & \multicolumn{1}{c}{ST (SDB)} & \multicolumn{1}{c}{ST (r=2)} \\ \midrule
npar & 11 & 11 & 13 & 12 & 12 & 14 \\
loglik & -778.56&	-791.63&	-777.72&	-737.92&	-754.28&	-735.51\\
AIC & 1579.12& 1605.26&	1581.44&	1499.84&	1532.56& \bf	1499.02\\ \bottomrule
\end{tabular}
\end{table}

Table \ref{tab:schizo1} presents the number of estimated parameters (npar), the maximum log-likelihood value attained (loglik), and the Akaike information criterion (AIC) for SN and ST distributions with $r=1$, $r=2$, and the particular case of SDB. 
The lowest AIC value is the one from the ST-LMM with $r=2$, closely followed by the ST model with $r=1$. 
Figure \ref{fig:schizoFits} presents contour plots of the estimated distribution of the random  effects and empirical Bayes estimates of the random effects for all the ST-LMM models considered. The model with $r=2$ seems to fit the data slightly better, but the fit from the model with $r=1$ seems reasonably similar. 

\begin{figure}[ht]\centering
    \centering
    \includegraphics[width=1\textwidth]{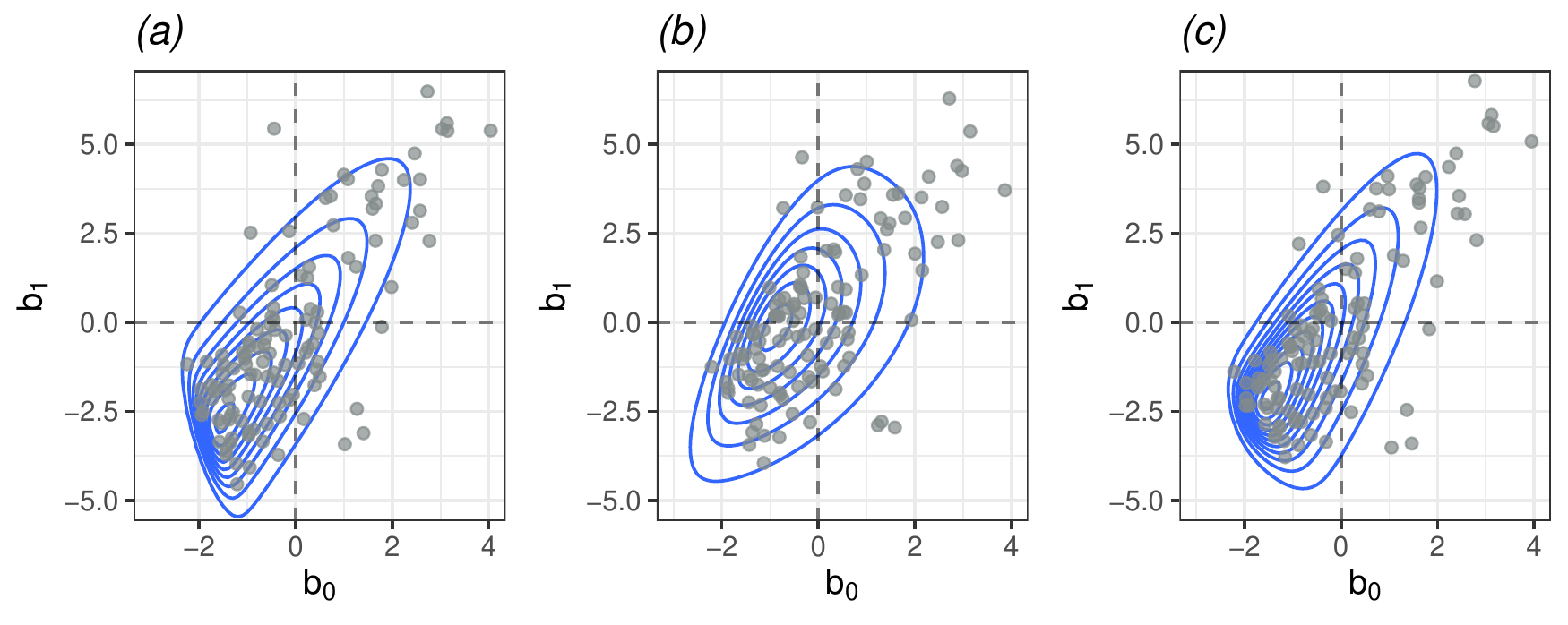}
    \caption{Contour plots of the estimated distribution of the random effects (blue curves), and empirical Bayes estimates of the random effects (gray dots), for the ST-LMM models fitted to the schizophrenia data set, with $r=1$ {\it (a)}, SDB {\it (b)}, and $r=2$ {\it (c)}.} 
\label{fig:schizoFits}
\end{figure}

For the sake of comparison, Table \ref{tab:schizo2} summarizes the results from ML estimation of the models with $r=1$ and $r=2$. 
Furthermore, analyzing the ST model with $r=2$ using the asymptotic normal approximation for the distribution of ML estimators with $\alpha=0.05$, we conclude that all fixed effects are significant, except the ones associated with the treatment effect ($\beta_3$ and $\beta_4$), corroborating with the equivalence hypothesis of the new antipsychotic drug.

\begin{table}[ht]
\centering
\caption{ML results from fitting the ST-LMM to the schizophrenia data set.}\label{tab:schizo2}
\begin{tabular}{@{}c|rr|rr@{}}
\toprule
\multirow{2}{*}{Parameter} & \multicolumn{2}{c|}{ST (r=1)} & \multicolumn{2}{c}{ST (r=2)} \\ \cmidrule(l){2-5} 
 & \multicolumn{1}{c}{Estimate} & \multicolumn{1}{c|}{SE} & \multicolumn{1}{c}{Estimate} & \multicolumn{1}{c}{SE} \\ \midrule
 $\beta_0$ & 2.668 & 0.136 & 2.658 & 0.142 \\
$\beta_1$ & -0.971 & 0.378 & -1.316 & 0.352 \\
$\beta_2$ & 6.722 & 0.507 & 6.712 & 0.514 \\
$\beta_3$ & -0.210 & 0.152 & -0.247 & 0.155 \\
$\beta_4$ & -0.267 & 0.468 & -0.099 & 0.444 \\
$\sigma^2$ & 0.213 &  & 0.199 &  \\
$\D_{11}$ & 0.098 &  & 0.061 &  \\
$\D_{12}$ & -0.304 &  & -0.092 &  \\
$\D_{22}$ & 1.431 &  & 0.560 &  \\
$\bDelta_{11}$ & 1.966 &  & 1.562 &  \\
$\bDelta_{21}$ & 3.544 &  & 3.324 &  \\
$\bDelta_{12}$ & - &  & 0.586 &  \\
$\bDelta_{22}$ & - &  & -0.852 &  \\
$\nu$ & 5.000 &  & 4.000 & \\
\bottomrule
\end{tabular}
\end{table}

\section{Final remarks}\label{sec:final}

This work developed a robust approach to relax the mathematical convenient normal assumptions usually considered in LMM, by considering a flexible formulation of the ST distribution that has as particular cases the proposals of \cite{azzalini2003distributions} and \cite{Sahu_Dey_Marcia}. The codes developed for estimation of the ST-LMM are available for download at the GitHub repository \url{https://github.com/fernandalschumacher/cfstlmm}.

Even though the general model formulation considered in this work ensures that the random effect and error are uncorrelated, they are not independent in general. In this regard, an interesting extension would be to consider different mixing variables for the random effect and the error, as in \cite{asar2018linear}. Nevertheless, in this case, the likelihood function has no closed-form, and therefore the use of approximated approaches, such as a Monte Carlo EM algorithm, is necessary.
Additionally, the proposed formulation can be easily extended to accommodate within-subject serial dependence, by considering useful structures such as damped exponential correlation \citep[DEC,][]{munoz1992parametric} or autoregressive correlation of order $p$  \citep[AR($p$),][]{box1976time}.

\section*{Acknowledgements}
This study was financed in part by the Coordenação de Aperfeiçoamento de Pessoal de Nível Superior - Brasil (CAPES) - Finance Code 001,  by the Conselho Nacional de Desenvolvimento Científico e Tecnológico - Brasil (CNPq), and by the Universidade Federal do Amazonas (UFAM).

\end{document}